\documentclass[11pt]{article}
\usepackage{fullpage}

\usepackage{graphics}
\usepackage[dvips]{epsfig}

\usepackage{amsmath}
\usepackage{amssymb}
\usepackage{amsfonts}
\usepackage{graphicx}

\usepackage{algorithm} 
\usepackage[noend]{algorithmic}
\usepackage{amsthm}
\usepackage{color}

\newtheorem{theorem}{Theorem}[section]

\newcommand{\remove}[1]{}

\begin{document}

\baselineskip  0.2in %  0.2in %0.18in si on veut compact
\parskip     0.1in %    0.1in % 0.0in  pour compacter
\parindent   0.0in %    0.0in % 0.3in pour voir les paragraphes

\title{{\bf Reaching a Target in the Plane with no Information}}

\author{
Andrzej Pelc\thanks{D\'{e}partement d'informatique, Universit\'{e} du Qu\'{e}bec en Outaouais,
Gatineau, Qu\'{e}bec J8X 3X7,
Canada. E-mail: pelc@uqo.ca.
Supported in part by NSERC discovery grant 8136 -- 2013 
and by the Research Chair in Distributed Computing of
the Universit\'{e} du Qu\'{e}bec en Outaouais.}
}

\date{ }
\maketitle

\begin{abstract}
A mobile agent has to reach a target in the Euclidean plane. Both the agent and the target are modeled as points. At the beginning, the agent is at distance at most $D>0$ from the target. Reaching the target means that the agent gets at a {\em sensing distance} at most $r>0$ from it. The agent has a measure of length and a compass. We consider two scenarios: in the {\em static} scenario the target is inert, and in the {\em dynamic} scenario it may move arbitrarily at any (possibly varying) speed bounded by $v$. The agent has no information about the parameters of the problem, in particular it does not know $D$, $r$ or $v$. The goal is to reach the target at lowest possible cost, measured by the total length of the trajectory of the agent.

Our main result is establishing the minimum cost (up to multiplicative constants) of reaching the target under both scenarios, and providing the optimal algorithm for the agent. For the static scenario the minimum cost is $\Theta((\log D + \log \frac{1}{r}) D^2/r)$, and for the dynamic scenario it is $\Theta((\log M + \log \frac{1}{r}) M^2/r)$, where $M=\max(D,v)$. Under the latter scenario, the speed of the agent in our algorithm grows exponentially with time, and we prove that for an agent whose speed grows only polynomially with time, this cost is impossible to achieve.

\vspace{2ex}

\noindent {\bf Keywords:} exploration, algorithm, mobile agent, speed, cost, plane. 

\vspace{7cm}
\end{abstract}

\thispagestyle{empty}
\setcounter{page}{0}
\pagebreak

%%%%%%%%%%%%%%%%%%%%%%%%%%%%%%%%%%%%%%%%%%%%%%%%%%%%%%%%%%%
\section{Introduction}
%%%%%%%%%%%%%%%%%%%%%%%%%%%%%%%%%%%%%%%%%%%%%%%%%%%%%%%%%%%

{\bf The background and the problem.}
Finding a target hidden in various environments is a task extensively studied in the literature. When the target is inert, this task is often called {\em treasure hunt}, and when it can move, it is called {\em pursuit}. In the latter case, the target is sometimes called the {\em robber} and the mobile agents chasing it are called {\em cops}.

The novelty of our formulation of this old problem is twofold. The environment is very simple, it is the Euclidean plane, and the mobile agent that has to find the target has no information whatsoever 
about the parameters of the problem. We ask what is the minimum cost, measured  by the total length of the trajectory of the agent to reach the target under this complete ignorance.

There is one mobile agent and the target, both  modeled as points. At the beginning, the agent is at distance at most $D>0$ from the target. Reaching the target means that the agent gets at a {\em sensing distance} at most $r>0$ from it. When this happens, we say that the agent {\em senses the target}. The agent has a measure of length and a compass showing the cardinal directions. We consider two scenarios: in the {\em static} scenario the target is inert, and in the {\em dynamic} scenario it may move arbitrarily at any (possibly varying) speed bounded by $v$. The agent has no information about the parameters of the problem, in particular it does not know $D$, $r$ or $v$.
This total ignorance may occur in many applications. If the target is a lost object, its position and even its distance from the agent may be unknown. If it is a hostile agent able to move, the chasing agent may not know any bound on its speed. Finally, the sensing distance at which the mobile agent can locate the target may also be unknown because it may depend on factors impossible to control. For example, if sensing is visual, the parameter $r$ may depend on the unknown density of the fog, and if it is chemical, it may depend on the unknown strength of the scent emitted by the target.

{\bf Our results.}
Our main result is establishing the minimum cost (up to multiplicative constants) of reaching the target under both scenarios, and providing the optimal algorithm for the agent. For the static scenario the minimum cost is $\Theta((\log D + \log \frac{1}{r}) D^2/r)$, and for the dynamic scenario it is $\Theta((\log M + \log \frac{1}{r}) M^2/r)$, where $M=\max(D,v)$. Under the latter scenario, the speed of the
chasing agent in our algorithm grows exponentially with time, and we prove that for a chasing agent whose speed grows only polynomially with time, this cost is impossible to achieve.

%--------------------------------------------------

%--------------------------------------------------
{\bf Related work.}
The problem of searching for a target by one or more mobile agents is an extensively studied task that was investigated under many different scenarios.
The environment where the target is hidden may be a graph or a plane, the search may be deterministic or randomized, and the target may be inert or mobile.
The book \cite{AG} surveys both the search for a fixed target and the related rendezvous problem, where the target and the finder are both mobile and
their role is symmetric: they both cooperate to meet. This book is concerned mostly with randomized search strategies. In \cite{MP,TSZ} the authors study relations between the problems of treasure hunt (searching for a fixed target) and rendezvous in graphs. By contrast, the survey \cite{CHI} and the book \cite{BN} consider pursuit-evasion games, mostly on graphs, where pursuers try to catch a fugitive trying to escape.  The authors of \cite{BCR} studied the task of finding a fixed point on the line and in the grid, and initiated the study of the task
of searching for an unknown line in the plane. This research was continued, e.g., in \cite{JL,La2}. In \cite{SF} the authors concentrated on game-theoretic aspects of
the situation where multiple selfish pursuers compete to find a target, e.g., in a ring. The main result of \cite{La} is an optimal algorithm to sweep a plane in order to locate an unknown fixed target, where locating means to get the agent originating at point $O$ to a point $P$ such that the target is in the segment $OP$. In \cite{FHGTM} the authors consider the generalization of the search problem in the plane to the case of several searchers.  To the best of our knowledge, the problem of efficient search for a fixed or a moving target in the plane, under complete ignorance of the searching agent, has never been studied before.

\section{The static scenario}

In this section we assume that the target is inert. Our algorithms produce a trajectory of the mobile agent which is a polygonal line whose segments are parallel to the cardinal directions. 
For any positive real $x$, the instruction $(N,x)$ (resp. $(E,x)$, $(S,x)$ , and $(W,x)$) has the meaning ``go North (resp. East, South, and West) at distance $x$''. Juxtaposition is used for concatenation
of trajectories, and $\overline{T}$ denotes the trajectory reverse with respect to trajectory $T$. For any positive real $y$, let $Q(y)$ denote the square with side $y$ centered at the starting point of the
mobile agent. 

For any positve integers $k$ and $j$, the {\em spiral} $S(k,j)$ is the trajectory resulting from the following sequence of instructions:
$(E,2^{-j})$, $(S,2^{-j})$, $(W,2\cdot 2^{-j})$, $(N,2 \cdot 2^{-j})$, $(E,3\cdot 2^{-j})$, $(S,3 \cdot 2^{-j})$, $(W,4\cdot 2^{-j})$, $(N,4 \cdot 2^{-j})$, ..., $(E,(2k+1)\cdot 2^{-j})$, $(S,(2k+1) \cdot 2^{-j})$, $(W,(2k+2)\cdot 2^{-j})$, $(N,(2k+2) \cdot 2^{-j})$. Note that, during the traversal of the spiral $S(k,j)$, the mobile agent  gets at distance less than $2^{-j}$ from every point of the square $Q(2k\cdot 2^{-j})$.
Denote by  $\Pi(k,j)$ the trajectory $S(k,j)\overline{S(k,j)}$. 

Consider the infinite matrix $A$ whose rows are numbered by consecutive positive integers and whose columns are numbered by consecutive positive even integers. The term $A(i,j)$ in row $i$ and column $j$ is the trajectory $\Pi(2^{i+j},j)$. For any positive integer $i$, denote by $\Delta[i]$ the concatenation $\Pi(2^{i+2},2)\Pi(2^{i+3},4),\dots \Pi(2^{1+2i},2i)$ of trajectories in the $i$th diagonal of the matrix.

Now the algorithm can be succinctly formulated as follows.

{\bf Algorithm} {\tt Static}

Follow the trajectory $\Delta[1]\Delta[2]\Delta[3]\dots$ until sensing the target.

\begin{theorem}\label{static}
The cost of Algorithm {\tt Static} is $O((\log D + \log \frac{1}{r}) D^2/r)$, where $D$ is an upper bound on the initial distance of the agent from the target and $r$ is the sensing distance.
\end{theorem}

\begin{proof}
Let $a=\lceil \log D \rceil$ and let $b$ be the smallest even integer greater or equal to $\lceil \log \frac{1}{r}\rceil$. Thus $b \leq \lceil \log \frac{1}{r}\rceil +1$.  The term $A(i,j)$ of the matrix $A$ is in the $(i+\frac{j}{2}-1)$th diagonal. During the traversal of the spiral $S(2^{a+b},b)$ the agent reaches the target. Hence it reaches the target while following the trajectory $\Delta[a+\frac{b}{2}-1]$.

The length of $\Pi(k,j)$ is $4(1+2+\cdots + (2k+2))2^{-j}=2(2k+2)(2k+3)2^{-j} \leq 40k^22^{-j}$. Hence the length of the trajectory $\Delta[i]$ is at most $40i2^{2i+2}$. The cost of the algorithm is at most the sum of lengths of trajectories $\Delta[1]$, $\Delta[2]$, ..., $\Delta[a+\frac{b}{2}-1]$, which is at most $80i2^{2i+2}$, where $i=a+\frac{b}{2}-1$. This  is $O((\log D + \log \frac{1}{r}) D^2/r)$.
\end{proof}

The following result shows that Algorithm {\tt Static} has cost optimal up to multiplicative constants.

\begin{theorem}\label{lb}
The cost of any algorithm for reaching a fixed target, with unknown bound $D$ on the initial distance and unknown sensing distance $r$, is at least 
 $\frac{1}{16}((\log D + \log \frac{1}{r}) D^2/r)$, for some couple of parameters $D$ and $r$, for which this value is arbitrarily large. 
\end{theorem}

\begin{proof}
First observe that  the cost of reaching a target in any polygon of area $P$, with sensing distance $r$, is at least $(P-\pi r^2)/(2r)$ because for any trajectory of length $x$, the set of points that are at distance at most $r$ from some point of the trajectory has area at most $2rx +\pi r^2$.

For any positive integer $i$, consider couples of integers $(D_j, r_j)$, where  $D_j=2^j$ and $r_j= 2^{-2(i-j+1)}$, for $j=1,2,\dots , i$.
Denote $R_1=Q(2)$ and $R_k=Q(2^{k})\setminus Q(2^{k-1})$, for $k>1$. The area of the polygon $R_j$ is larger than $D_j^2/2$, and all points of it are at distance smaller than $D_j$ from the origin of the mobile agent. Suppose that the target is hidden in polygon $R_j$, with sensing distance $r_j$, where $j=1,2,\dots , i$.  Hence the cost of any algorithm for reaching the target in $R_j$ is at least $\frac{D_j^2}{4r_j}-\frac{\pi r_j}{2} \geq \frac{D_j^2}{8r_j} = \frac{1}{2}2^{2i}$.
Since polygons $R_j$ are pairwise disjoint, the cost of any algorithm for reaching the target in every polygon $R_j$, with  sensing distance $r_j$, is at least 
$$\frac{1}{2}i2^{2i}=\frac{1}{8}(\log D_j+ \frac{1}{2}\log \frac{1}{r_j}-1)\frac{D_j^2}{r_j}\geq \frac{1}{16}(\log D_j+ \log \frac{1}{r_j})\frac{D_j^2}{r_j} .$$
\end{proof}

\section{The dynamic scenario}

In this section we consider the scenario of a target moving arbitrarily at possibly varying speed bounded by some constant $v$. Neither the bound  $D$ on the initial distance to the target, nor the sensing distance $r$, nor the bound $v$ on the speed of the target are known to the mobile agent.

The proposed algorithm is a simple modification of Algorithm {\tt Static} that specifies the speeds at which the mobile agent traverses the consecutive diagonals.

{\bf Algorithm} {\tt Dynamic}

Follow the trajectory $\Delta[1]\Delta[2]\Delta[3]\dots$, using speed $2^{5i}$ on trajectory $\Delta[i]$,  until sensing the target.

\begin{theorem}
The cost of Algorithm {\tt Dynamic} is $O((\log M + \log \frac{1}{r}) M^2/r)$, where $M=\max(D,v)$,  $D$ is an upper bound on the initial distance of the agent from the target, $v$ is an upper bound on the speed of the moving target, and $r$ is the sensing distance.
\end{theorem}

\begin{proof}
As calculated in the proof of Theorem \ref{static}, the length of the trajectory $\Delta[i]$ is at most $40i2^{2i+2}$. Hence it is at most $2^{3i}$, for $i \geq 11$.
Since the trajectory $\Delta[i]$ is traversed by the agent at speed $2^{5i}$, it follows that the time $t_i$ of traversing  $\Delta[i]$ is at most $2^{-2i}$, for $i \geq 11$,
and consequently the total time of traversing all trajectories $\Delta[i]$ is bounded by some constant $q$.

Let $a=\lceil \log D \rceil$ and let $b$ be the smallest even integer greater or equal to $\lceil \log \frac{1}{r}\rceil$. Thus $b \leq \lceil \log \frac{1}{r}\rceil +1$. 
Let $c= \lceil qv \rceil$. Let $D'=D +qv$. This is the maximum distance from the starting point of the mobile agent that the target can get. Let $a'=\max(a,c)+1$. Hence $D' \leq 2^{a'}$.
The term $A(a',b)$ is in the $y$th diagonal of the matrix $A$, where $y=a'+\frac{b}{2}-1$. If $t_y \leq \frac{1}{v\cdot 2^{b+1}}$, then the target is reached by the time the agent traverses the trajectory $\Delta_y$.
Indeed, if this inequality is satisfied, the agent traverses the spiral $S(2^{a'+b},b)$ faster than the target can traverse the distance $\frac{1}{2^{b+1}}$, and hence it must get at distance at most $r$ from the target during the traversal of this spiral. Since $2y\geq b+c$, we get $t_y\leq \frac{1}{2^{2y}}\leq \frac{1}{2^{b+c}}\leq  \frac{1}{v\cdot 2^{b+1}}$, for $y \geq 11$. 

If $y\leq 10$ then the cost of the algorithm is $O(1)$, hence we may assume  $y \geq 11$. In this case 
the cost of the algorithm is at most the sum of lengths of trajectories $\Delta[1]$, $\Delta[2]$, ..., $\Delta[a'+\frac{b}{2}-1]$, which is at most $80y2^{2y+2}$, where $y=a'+\frac{b}{2}-1$.
This is $O((\log M + \log \frac{1}{r}) M^2/r)$.
\end{proof}

The following result shows that Algorithm {\tt Dynamic} has cost optimal up to multiplicative constants, if the adversary can choose the parameters $D$, $r$ and $v$.

\begin{theorem}
The cost of any algorithm for reaching a moving target, with unknown bound $D$ on the initial distance, unknown bound $v$ on the speed of the target, and unknown sensing distance $r$, is at least 
 $c((\log M + \log \frac{1}{r}) M^2/r)$, where $c$ is some positive constant and $M=\max(D,v)$, for some triple of parameters $D$, $v$ and $r$, for which this value is arbitrarily large. 
\end{theorem}

\begin{proof}
If $D \geq v$, then the theorem follows immediately from Theorem \ref{lb}, with the target choosing speed $0\leq v$, i.e., being inert.
Hence we may assume that $v>D$. It is enough to show the lower bound $c((\log v + \log \frac{1}{r}) v^2/r)$ on the cost of any algorithm for reaching the target, for some constant $c>0$, when the initial distance is $D=1$. Let $O$ be the starting point of the mobile agent and let $P$ be the starting point of the moving target.
Let $S(O,x)$ denote the disc of radius $x$ centered at~$O$.

%Denote $C(\delta,r)=\frac{1}{8}((\log \delta + \log \frac{1}{r}) \delta^2/r)$. Let $k$ be any positive integer and let $\delta_k$ be such that $C(\delta_k,1) \geq k$. 

Consider any algorithm for the mobile agent, with speed $f(t)$ at time $t$. Let $t_0$ be the constant for which $\int_0^{t_0}f(t)dt=\frac{1}{2}$. During the time interval 
$[0,t_0]$ the mobile agent must be inside the disc $S(O,\frac{1}{2})$. Suppose that, during the time interval 
$[0,t_0]$, the target moves at speed at most $v$ on the line $OP$ away from 
the point $O$. Since the point $P$ can be anywhere  at distance $1$ from $O$, it follows that at time $t_0$ the target can be at any point of the set 
$S(O,1+vt_0) \setminus S(O,1)$. In particular, it can be at any point $Q$ of a square $Z$ of size $\frac{1}{2}vt_0$ contained in this set. From time $t_0$ on, the adversary keeps the target inert at point $Q$. Thus the cost of the algorithm is at least the cost of the best
algorithm in the static scenario, where the inert target is in $Z$. 

Since the area of the square $Z$ is $\frac{t_0^2}{4}v^2$, by a similar argument as in the proof of the Theorem \ref{lb}, we get that this cost is at least 
$\frac{t_0^2}{128}((\log v + \log \frac{1}{r}) v^2/r)$, for some sensing distance $r$ and some bound $v$ on the speed of the target that make this value arbitrarily large. 
\end{proof}

While Algorithm {\tt Dynamic} has been proved to have cost optimal up to multiplicative constants, its downside is the fact that the speed of the mobile agent
grows exponentially with time. Hence it is natural to ask if the same cost can be achieved with the speed of the mobile agent growing only polynomially. Our final result shows that this is impossible.

\begin{theorem}
There does not exist an algorithm for reaching a moving target at cost optimal up to multiplicative constants, with the speed of the mobile agent growing polynomially in time.
\end{theorem}

\begin{proof}
Let $O$ be the starting point of the mobile agent and let $P$ be the starting point of the moving target. Suppose that the distance between $O$ and $P$ is $D=1$
and that the target moves at speed at most $v\geq 1$ on the line $OP$ away from 
the point $O$. Let $S(O,x)$ denote the disc of radius $x$ centered at~$O$.
Suppose that there exists an algorithm for reaching the target at cost at most $C(v,r)=d((\log v + \log \frac{1}{r}) v^2/r)$, for some constant $d>0$, where $r$ is the sensing distance, and the mobile agent
moves at speed at most $t^c$ at time $t$, for some positive integer constant $c$.

Let $t$ be the earliest time at which the agent reaches the target. At this time the target can be anywhere in the set $Z=S(O,1+vt)\setminus S(O,1)$, and the trajectory of the agent has length at most $\int_0^t t^cdt=\frac{t^{c+1}}{c+1}$.
Since the area of the set $Z$ is $\pi (1+vt)^2-\pi>\pi(vt)^2$, we must have $\frac{t^{c+1}}{c+1} \geq  \frac{\pi (vt)^2-\pi r^2}{2r}>\frac{(vt)^2}{2r}$, for sufficiently large $v$ and sufficiently small $r$. Thus we must have $\frac{t^{c-1}}{c+1}>\frac{v^2}{2r}$, which
implies $t>(\frac{(c+1)v^2}{2r})^{\frac{1}{c-1}}$. It follows that the cost  of the algorithm at time $t$ is larger than $\frac{1}{c+1}(\frac{(c+1)v^2}{2r})^{\frac{c+1}{c-1}}$, which is $\alpha (\frac{v^2}{r})^{\beta}$,
for some constants $\alpha>0$ and $\beta>1$. However, we have $C(v,r)=d((\log v + \log \frac{1}{r}) v^2/r)\leq d (\log \frac{v^2}{r})\frac{v^2}{r}<\alpha (\frac{v^2}{r})^{\beta}$, for sufficiently large $v$ and sufficiently small $r$, which is a contradiction.
\end{proof}

\section{Conclusion}

We established the optimal cost (up to multiplicative constants) of reaching a target in the plane by a mobile agent ignorant of the parameters of the problem,  which are a bound $D$ on the initial distance from the target, the sensing distance $r$, and a bound  $v$ on the speed of the target. 

Our assumption about the target was that it moves at a bounded (although possibly varying) speed, and our assumption about the mobile agent was that it has a compass and a measure of length.
Is it possible to weaken or eliminate these assumptions? Notice that without any assumption on the speed of the target, the adversary could assign it the same speed as that of the agent at all times (which is possible if the agent does not know anything about the speed of the target) and make it go away from the agent, (on the extension of the segment between the initial positions of the agent and of the target) which clearly precludes reaching the target. As for the assumption about the capabilities
of the agent, having a measure of distance seems necessary, as otherwise no meaningful algorithm to move in the plane can be formulated. On the other hand, the availability of a compass could be removed. In our algorithms
the only needed features are the possibility of going straight at some chosen distance in a chosen direction, and forming right angles in the plane. Thus, e.g., an arbitrarily distorted compass with a fixed distortion,
e.g., always showing South-East as North, would be sufficient. In fact, since drawing right angles is possible with an available measure of distance, no compass is needed, as long as the agent  has the ability of going straight (which is difficult in dense fog, but easy on a clear day on the snow).

%%%%%%%%%%%%%%%%%%%%%%%%%%%%%%%%%%%%%%%%%%%%%%%%%%%%%%%%%%%
\bibliographystyle{plain}

%%%%%%%%%%%%%%%%%%%%%%%%%%%%%%%%%%%%%%%%%%%%%%%%%%%%%%%%%%% 

\end{document}